\documentclass[final]{alt2023}
\usepackage{tikz-cd}
\usepackage{tikz-network}
\usepackage{mathdots}
\usepackage{csquotes}
\usepackage{bm}
\usepackage{xargs}
\usepackage{comment}
\usepackage{standalone}
\usepackage{sgame}

\tikzset{ball/.style={circle, draw, fill=black,inner sep=0pt, minimum width=4pt}}
\tikzset{nd/.style={inner sep=1pt}}
\tikzset{CRS/.style={circle, draw,inner sep=1pt, minimum width=8pt}}

\usepackage{pgfplots}
\pgfplotsset{compat = newest}

\usepackage{xcolor}
\usetikzlibrary{decorations.markings}
\usetikzlibrary{arrows.meta}
\tikzset{>=Latex}
\tikzset{
  set arrow inside/.code={\pgfqkeys{/tikz/arrow inside}{#1}},
  set arrow inside={end/.initial=>, opt/.initial=},
  /pgf/decoration/Mark/.style={
    mark/.expanded=at position #1 with
    {
      \noexpand\arrow[\pgfkeysvalueof{/tikz/arrow inside/opt}]{\pgfkeysvalueof{/tikz/arrow inside/end}}
    }
  },
  arrow inside/.style 2 args={
    set arrow inside={#1},
    postaction={
      decorate,decoration={
        markings,Mark/.list={#2}
      }
    }
  },
}

\newtheorem{thm}{Theorem}[section]

\newtheorem{defn}[thm]{Definition}

\newtheorem{lem}[thm]{Lemma}
\newtheorem{corol}[thm]{Corollary}

\newtheorem{conj}[thm]{Conjecture}
\newtheorem{prop}[thm]{Proposition}

\newcommand{\real}{\mathbb{R}}

\newcommand{\exptu}{\mathbb{U}}
\newcommand{\dist}{\mathbf{d}}

\newcommand{\arc}[3][]{\begin{tikzcd} #2 \ar[r,-Latex,"#1"] \pgfmatrixnextcell #3 \end{tikzcd}}

\newcommand{\chain}[3][]{\begin{tikzcd} #2 \ar[r,dashed,-Latex,"#1"] \pgfmatrixnextcell #3 \end{tikzcd}}
\newcommand{\chaineq}[3][]{\begin{tikzcd} #2 \ar[r,dashed,dash,"#1"] \pgfmatrixnextcell #3 \end{tikzcd}}

\DeclareMathOperator{\cl}{cl}

\DeclareMathOperator{\intr}{int}
\DeclareMathOperator{\content}{content}

\altauthor{\Name{Oliver Biggar} \Email{oliver.biggar@anu.edu.au}
\AND
  \Name{Iman Shames} \Email{iman.shames@anu.edu.au}\\
  \addr CIICADA Lab, Australian National University, Canberra, 2601, Australia}

\title[Response Graphs and Chain Components]{The Replicator Dynamic, Chain Components and the Response Graph}

\begin{document}

\maketitle

\begin{abstract}
    In this paper we examine the relationship between the flow of the replicator dynamic, the continuum limit of Multiplicative Weights Update, and a game's \emph{response graph}. We settle an open problem establishing that under the replicator, \emph{sink chain components}---a topological notion of long-run outcome of a dynamical system---always exist and are approximated by the \emph{sink connected components} of the game's response graph. More specifically, each sink chain component contains a sink connected component of the response graph, as well as all mixed strategy profiles whose support consists of pure profiles in the same connected component, a set we call the \emph{content} of the connected component. As a corollary, all profiles are chain recurrent in games with strongly connected response graphs. In any two-player game sharing a response graph with a zero-sum game, the sink chain component is unique. In two-player zero-sum and potential games the sink chain components and sink connected components are in a one-to-one correspondence, and we conjecture that this holds in all games.
\end{abstract}

\section{Introduction}

When a collection of players simultaneously learn a game, what strategies do they eventually learn to play? This is the fundamental question of evolutionary game theory, and by extension is of critical importance in economics~\citep{sandholm2010population}, biology~\citep{smith1973logic}, and computer science~\citep{roughgarden2010algorithmic}. In particular, this question lies at the heart of \emph{multi-agent learning}~\citep{yang2020overview} which itself is at the core of recent breakthroughs in AI~\citep{silver2016mastering,silver2017mastering,silver2018general}. 

Unpacking this question in more depth, we find three key observations. The first: the question assumes a fixed learning or evolution algorithm, which in evolutionary game theory is called a \emph{dynamic}. When paired with the strategy space, we obtain a \emph{dynamical system}~\citep{strogatz2018nonlinear}. We expect dynamical systems theory to play a role in the solution. The second: as the word `eventually' suggests, we are interested in \emph{long-run} or asymptotic analysis of the dynamical system. The third: learning is a computational process, and so our analysis must be compatible with computational feasibility.

Despite our game-theoretic intuition, \emph{Nash equilibria}~\citep{nash1951non} are generally not the answer to our question. A series of results in evolutionary game theory~\citep{kleinberg2011beyond,sandholm2010population,papadimitriou2019game} have shown that standard game dynamics generally do not converge to (mixed) Nash equilibria~\citep{vlatakis2020no}\footnote{By contrast, many dynamics do converge to \emph{pure} Nash equilibria.}, even in zero-sum games~\citep{mertikopoulos2018cycles}. In fact, no choice of dynamic can converge to Nash equilibria in all games~\citep{hart2003uncoupled,benaim2012perturbations,milionis2022nash}. From a computational perspective, there are also problems: Nash equilibria are PPAD-complete to compute~\citep{daskalakis2009complexity}, and so our players cannot be expected to feasibly learn them.

If the Nash equilibrium is not the right answer, we must determine what \emph{is}. Using our key ideas, we set out three criteria for a solution concept, inspired by~\cite{fabrikant2008complexity}:
\begin{enumerate}
    \item (\textbf{Convergence}) Almost all points should converge to the solution concept, and should remain there once reached.
    \item (\textbf{Existence}) The solution concept should exist in all games.
    \item (\textbf{Computability}) The solution concept should be efficiently computable.
\end{enumerate}

Finally, to avoid trivial solutions\footnote{Taking the entire strategy space of the game as the solution concept satisfies these criteria, but it is not a particularly enlightening solution.}, we desire the \emph{minimal} set of strategy profiles satisfying these properties---\textbf{Existence}, in particular. A priori, these criteria---particularly the first!---seem very optimistic. We restrict our attention to the best-studied continuous-time dynamic, the \emph{replicator}~\citep{sandholm2010population}. The replicator emerged originally from population models in biology, and is the continuum limit of the \emph{Multiplicative Weights Update} algorithm~\citep{arora2012multiplicative}.
Yet even this well-studied and comparatively well-behaved dynamic's convergence property is notoriously difficult to understand, with general results only known for a few classes of games (such as zero-sum and potential games, as in~\citeauthor{mertikopoulos2018cycles}).

But hope remains, if we use the right mathematical tools. \cite{papadimitriou2019game,papadimitriou2018nash,papadimitriou2019game} recently proposed a new solution concept meeting the \textbf{Convergence} criterion, inspired by the concept of \emph{chain recurrence}, a dynamical systems concept underlying the Fundamental Theorem of Dynamical Systems~\citep{conley1978isolated}. To understand this concept, we first note that our \textbf{Convergence} criterion breaks into two parts, (1) converging to the solution concept, and (2) remaining there once reached. An example of a point satisfying the latter is a \emph{stationary point}, a point which remains in place under the dynamic. A more general notion of `staying in place' is a \emph{periodic} point, one which returns to itself repeatedly. A yet more general concept is a \emph{recurrent} point, one which returns arbitrarily close to itself infinitely often. A major historical challenge of dynamical systems was to find the appropriate generalisation of a recurrent point, such that all points `end up at' these stationary ones~\citep{alongi2007recurrence}. This was achieved by \citet{conley1978isolated}, who introduced the notion of \emph{chain recurrent} points.

Conley's insight was to use a `noisy' generalisation of an orbit of the system. He introduced the concept of an $(\epsilon,T)$-chain, which is an orbit which allows for a finite number of tiny `jumps' of size at most $\epsilon$. To ensure we cannot jump `too often', each jump must be separated by a time of at least $T>0$. If there exists an $(\epsilon,T)$-chain from $x$ to $y$ for \emph{any} $\epsilon > 0$ and $T>0$, we say there is a \emph{pseudo-orbit} from $x$ to $y$. A point $x$ is \emph{chain recurrent} if there is a pseudo-orbit from $x$ to itself. Chain recurrent points form connected components of the space, called \emph{chain components}. The Fundamental Theorem of Dynamical Systems~\citep{conley1978isolated} establishes that all points converge to chain components. When we view the dynamical system from the perspective of pseudo-orbits rather than `true' orbits, we obtain a partial ordering on chain components. The minimal elements in this order are called \emph{sink chain components}. Points in these components satisfy both parts of the \textbf{Convergence} criterion: they are chain recurrent, so they remain in the component once reached, and because they are \emph{sink} chain components other points end up there under pseudo-orbits.

The definition of chain recurrence is fundamentally \emph{computational}. Any computational device with arbitrarily large but finite precision cannot distinguish between a `true' orbit of the dynamical system and a pseudo-orbit. In fact, many common game dynamics exhibit chaotic behaviour, demonstrating the difficulty in tracking `true' orbits~\citep{cheung2019vortices,cheung2020chaos,andrade2021learning}. Consequently, any solution concept which is `computational' in this sense cannot distinguish between points in the same chain component, making sink chain components the minimal set satisfying our \textbf{Convergence} and \textbf{Computability} criteria.

Unfortunately, despite their computational inspiration, actually computing the sink chain components is not obviously feasible (failing \textbf{Computability}). There is a second problem: the order defined by pseudo-orbits is infinite in general, and so sink chain components need not exist (failing \textbf{Existence}). \cite{papadimitriou2019game} solve the second problem by assumption: they conjecture that \emph{in game dynamics, sink chain components always exist}, which we prove for the replicator in this paper. They address the first problem by suggesting a computable surrogate of sink chain components: the \emph{sink connected components} of the game's \emph{response graph}.

The response graph is a directed graph defined on the pure profiles of the game. There is an arc between profiles if they differ in the strategy of a single player, with the arcs directed toward the preferred profile for that player. The response graph can be thought of as underlying combinatorial structure of the game, capturing precisely the order in which each player prefers their strategies given fixed choices of strategy for each other player. The response graph is \emph{structurally stable}; changing the payoffs of the game in a small way generally does not affect it. Response graphs also underlie \emph{ordinal games}~\citep{mertens2004ordinality} and \emph{strategic games}~\citep{candogan2011flows}, and so are a unifying model of the game structure. In spite of their generality, response graphs store many structural features of a game~\citep{biggar2022twoplayer}. As a solution concept, sink connected components can be easily computed by traversing the response graph (\textbf{Computability}), and always exist because the connected components of a graph is a \emph{finite} partial order (\textbf{Existence}). Restricted to the edges of the graph, the replicator dynamic flows in the direction of higher payoff, suggesting that this discrete structure will be a useful analogue of the replicator flow over the whole game. If the sink connected components also capture the \textbf{Convergence} properties of the sink chain components, then they meet all of the desired criteria for a solution concept.

Sink connected components are a simple combinatorial object which depend only on each player's preference order. By contrast, the sink chain components are complex topological objects which depend on payoffs and the dynamic. Yet, remarkably, we show that the sink connected components capture key properties of the sink chain components of the replicator. This is the topic of this paper:
\textbf{the connection between sink chain components of games under the replicator dynamic and the sink connected components of the response graph.}

\subsection{Contributions}

We show, firstly, that \emph{sink chain components always exist under the replicator dynamic} (Theorem~\ref{scc theorem}), as conjectured by \cite{papadimitriou2019game}. The proof uses the structure of the response graph; specifically, each attractor contains an attracting (in the sense of paths) set of nodes in the response graph. In the same theorem, we prove that \emph{sink chain components of the replicator contain sink connected components}. A weaker result is proved in~\cite{omidshafiei2019alpha}, which establishes that \emph{asymptotically stable} sink chain components, if they exist, contain sink connected components of the replicator. We believe our proof is the first to establish the general result for the replicator (without the asymptotic stability requirement)\footnote{A more general result is stated in~\cite[Theorem~4.1]{papadimitriou2019game}, but no proof is given and the statement is not true in general.
}

To extend this result, we define the \emph{content} of a sink connected component as the set of all mixed strategy profiles where all pure profiles in their support are in the sink connected component. We then prove that \emph{the content of a connected component is always contained in the associated chain component} (Theorem~\ref{content containment}). The content of a strongly connected response graph is all mixed profiles; as a corollary, we obtain the surprising result that all profiles are chain recurrent in any game with a strongly connected response graph (such as Matching Pennies or Rock-Paper-Scissors). Augmenting this result, we prove that \emph{if a sink connected component is a subgame, then the associated sink chain component is precisely that subgame} (Corollary~\ref{sink subgames}).

We then analyse the influence of the zero-sum and potential properties on the chain components. Because we are interested in the sink connected components, and thus the graph structure, we study all games whose response graphs are isomorphic to a zero-sum or potential game respectively; these are called \emph{preference-zero-sum} and \emph{preference-potential} games~\citep{biggar2022twoplayer}. \emph{The sink chain components of preference-potential games are exactly the pure Nash equilibria} (Theorem~\ref{sinks in preference potential}). Preference-zero-sum games have a unique sink connected component, and consequently have exactly one sink chain component (Lemma~\ref{lem: unique}).
This result is surprisingly analogous to the well-known result that two-player zero-sum games have a convex set of Nash equilibria. However, preference-zero-sum games are a much more general set of games~\citep{biggar2022twoplayer}, defined entirely by their graph structure.

Using our results, we show in Section~\ref{sec: applications} that the content of sink connected components completely characterises, and thus allows us to compute, the replicator sink chain components of \emph{all} $2\times n$ strict games, all-but-one $3\times 3$ strict preference-zero-sum games, and all games where every sink connected component is a subgame, such as preference-potential and weakly acyclic games~\citep{young1993evolution}.

These results suggest that the response graph has a much more significant impact on the outcome of the game than we might expect. We conjecture that there is always a one-to-one correspondence between sink chain components and sink connected components, suggesting that the long-run behaviour of the replicator dynamic is fundamentally governed by the graph structure of the game.

Wherever a result has no reference, a proof can be found in the Appendix.

\section{Related Work}

Though the idea of sink chain components as a tool to analyse games has appeared historically in the game theory literature~\citep{akin1984evolutionary}, the movement towards dynamic solution concepts has accelerated in the modern algorithmic game theory community~\citep{sandholm2010population}. Our solution concept criteria are derived from those in~\cite{fabrikant2008complexity}. \cite{kleinberg2011beyond} demonstrate that the replicator generically does not converge to Nash, and give an example where the game converges to a cycle with strictly higher social welfare than the unique Nash equilibrium. \cite{papadimitriou2016nash,papadimitriou2018nash} present the argument for using chain recurrence to analyse games, and study the replicator dynamic on zero-sum and weighted potential games as a demonstration. This was generalised to Follow-The-Regularised-Leader (FTRL) dynamics and network zero-sum games in~\cite{mertikopoulos2018cycles}. Our work most closely follows \cite{papadimitriou2019game}, where the authors argue for sink chain components and suggest a connection with sink connected components of the response graph. The authors also set out the importance of further understanding of sink chain components. Since then, \cite{vlatakis2020no} demonstrated that mixed Nash are never attracting under FTRL dynamics and argued for greater understanding the long-run behaviour of these dynamics. 

Following~\cite{papadimitriou2019game}, \cite{omidshafiei2019alpha} use the replicator, sink chain components and sink connected components to develop a method, called \emph{$\alpha$-rank}, to evaluate the strength of algorithms in multi-agent learning settings, such as AlphaGo~\citep{silver2016mastering}. The validity of $\alpha$-rank as a ranking method for algorithms is \emph{predicated on the premise} that sink connected components are a good surrogate for sink chain components of the replicator. In \cite{omidshafiei2019alpha}, this premise rests on a proof that \emph{asymptotically stable} sink chain components of the replicator are finite in number and contain sink connected components. It is not established that any sink chain components (asymptotically stable or not) exist. In this paper, we prove that sink chain components of the replicator always exist and always contain sink connected components (dropping the requirement that they be asymptotically stable)
. We also conjecture there is a one-to-one correspondence between sink chain and connected components, which if true would greatly strengthen the motivation for $\alpha$-rank.

The response graph is a concept of increasing interest, particularly in algorithmic game theory~\citep{fabrikant2008complexity,goemans2005sink,kleinberg2011beyond}. A labelled form of the response graph is a key component of the decomposition results of~\cite{candogan2011flows}. \cite{biggar2022twoplayer} explicitly the study the response graph and its sink connected components, and examine the influence of the zero-sum and potential properties on the graph. We use these results extensively in Section~\ref{sec: applications}. More recently, the structure of the response graph is used to describe the `landscape' of games, for the purposing of analysing multi-agent learning~\citep{omidshafiei2020navigating}.

\section{Preliminaries} \label{sec: preliminaries}

A game is a triple consisting of $N$ players, strategy sets $S_1,S_2,\dots,S_N$ and a utility function $u : \prod_{i=1}^{N}S_i \to \real^N$. We assume each strategy set $S_i$ is finite. An element of $\prod_{i=1}^{N}S_i$ (an assignment of strategies to players) we call a \emph{profile}, and we denote the set of profiles by $Z$. We use the notation $\bar p_{-i}$ to denote an assignment of strategies to all players other than $i$, which we call an \emph{antiprofile}, and we denote the set of all $i$-antiprofiles by $\bar Z_{-i}$. If we insert a strategy $s\in S_i$ in the $i$th index of $\bar p_{-i}$ we obtain a profile, and we denote this operation by $(s ; \bar p_{-i})$.
A \emph{subgame} of a game $([N],\{S_1,\dots,S_n\},u)$ is a game $([N],\{T_1,\dots,T_n\},u')$ where for each $i$, $T_i\subseteq S_i$, and $u'$ is $u$ restricted to $\prod_{i=1}^N T_i$.

Two profiles are \emph{$i$-comparable} if they differ only in the strategy of player $i$; they are \emph{comparable} if they are $i$-comparable for some player $i$. If two profiles are comparable, then there is exactly one $i$ such that they are $i$-comparable. We say a game is \emph{strict} if the payoffs to player $i$ in two $i$-comparable profiles are never equal.
The \emph{response graph} of a game $u:Z\to\real^N$ is the graph $\mathcal{G}_u = (Z, A)$ where there is an arc $\arc{p}{q}\in A$ between profiles $p$ and $q$ if and only if they are $i$-comparable and $u_i(p) \leq u_i(q)$. A subgraph of the response graph is \emph{attracting} if there are no paths out of it. The sink connected components are minimal attracting subgraphs.

A \emph{mixed strategy} is a distribution over a player's pure strategies, and a \emph{mixed profile} is an assignment of a mixed strategy to each player. We sometimes refer to a profile as a \emph{pure profile} to distinguish it from a mixed profile. For a mixed profile $x$, we write $x^i$ for the distribution over player $i$'s strategies, and $x^i_s$ for the $s$-entry of player $i$'s distribution, where $s\in S_i$. The set of mixed profiles on a game is given by $\prod_{i=1}^N \Delta_{|S_i|}$ where $\Delta_{|S_i|}$ are the simplices in $\real^{|S_i|}$. We denote $\prod_{i=1}^N \Delta_{|S_i|}$ simply by $X$, and call $X$ the \emph{strategy space} of the game. The utility function $u$ of a game extends naturally to mixed profiles. The \emph{expected utility function} of $u$ is $\exptu: X \to \real^N$, where \[
\exptu(x) = \sum_{q\in Z}u(q) \prod_{i=1}^N x^i_{q_i}
\]

\subsection{Dynamical Systems} \label{sec: dynamical systems}

We study the \emph{replicator dynamic}, which is a continuous-time dynamical system~\citep{sandholm2010population,hofbauer2003evolutionary} defined by the following ordinary differential equation, where $N$ is the number of players, $\exptu$ is the expected utility function, and $x$ is a mixed profile.
\[
\dot x_s^p = x_s^p\left(\exptu_p(s; \bar x_{-p}) - \sum_{t\in S_p} x_t^p \exptu_p(t; \bar x_{-p}) \right)
\]
The solutions to this equation define a \emph{flow}~\citep{sandholm2010population,vlatakis2020no} on the strategy space of a game, which is a function $\phi:X\times \real\to X$ which is a continuous group action of the reals on $X$. We call this the \emph{replicator flow}.
The forward orbit of the flow from a given point is called a \emph{trajectory} of the system. Flows are \emph{invertible}, that is, $\phi^{-1}$ is also a flow, called the \emph{time-reversed} flow. We make use of two special properties of the replicator dynamic. The first is that the replicator is subgame-independent: the support of a point is invariant along an orbit, and the trajectory is only defined by the payoffs in that subgame (Theorem~\ref{subgame invariance}). The second property of the replicator is \emph{volume preservation}: after a differentiable change of variables, the replicator preserves the volume of all sets on the interior of a game~\citep{akin1984evolutionary,hofbauer1996evolutionary,eshel1983coevolutionary,selten1988note,sandholm2010population,vlatakis2020no}. Consequently, no attractor or repellor (Definition~\ref{def: attractor}) can exist in the interior of the state space (Theorem~\ref{no interior attractors}).
\begin{thm}[Subgame-independence of the replicator] \label{subgame invariance}
Let $X$ be the strategy space of a game $u$, and $Y$ be the strategy space of a subgame $u'$ of $u$. The flow $\phi_u|_Y$ of the replicator on $u$ restricted to $Y$ is identical to $\phi_{u'}$, the replicator flow on $u'$.
\end{thm}
The fact that replicator trajectories have constant support is well-known~\citep{sandholm2010population}. The fact that the flow is defined by the payoffs for strategies in the support follows easily from the fact that all other terms in the differential equation vanish. This result allows us to analyse the flow of the replicator on a subgame of a game using induction on subgames, as we do in the proof of Theorem~\ref{content containment}. Another important dynamical systems concept are \emph{attractors}.

\begin{defn}[\citet{sandholm2010population}] \label{def: attractor}
Let $A$ be a compact, non-empty invariant set under a flow $\phi$ on a compact space $X$. If there is a neighbourhood $U$ of $A$ such that \[
\lim_{t\to\infty} \sup_{x\in U} \inf_{y \in A} \dist(\phi(x,t), y) = 0
\]
where $\dist$ is a metric, then we call $A$ an \emph{attractor}. An attractor of the time-reversed flow $\phi^{-1}$ we call a \emph{repellor}.
\end{defn}
There are many equivalent ways of defining an attractor~\citep{sandholm2010population}. In particular, a compact set $A$ being an attractor is equivalent to requiring that there is an open forward-invariant set $B$ with $\phi(\cl(B),t) \subset B$ for all times $t\geq T > 0$, and $A = \bigcap_{t\geq 0} \phi(B,t)$. Such a $B$ is called a \emph{trapping region} for $A$. Each attractor has a \emph{dual repellor}, defined by trapping regions.

\begin{lem}[\cite{sandholm2010population}]
Let $A$ be an attractor, with $B$ a trapping region for $A$. Then $A^* := \bigcap_{t\leq 0} \phi(X\setminus B,t)$ is a repellor, which we call the \emph{dual repellor of $A$}.
\end{lem}

Attractors and repellors are dual in the sense that $A^*$ is an attractor of the time-reversed flow $\phi^{-1}$, and in this flow $A$ is its dual repellor. On a compact space, the dual repellor is non-empty.
Now we can formally establish that the replicator has no attractors in its interior. 

\begin{thm}[The replicator has no interior attractors, \cite{eshel1983coevolutionary}] \label{no interior attractors}
Let $G$ be a game, with $X$ the strategy space. Under the replicator, there are no attractors in the interior of $X$.
\end{thm}

This does not mean that there are no attractors at all---only that each must contain at least some points on the boundary of $X$. This theorem also applies identically to repellors. Inductive repetition of this theorem combined with the flow along edges gives us:
\begin{lem}[Attractors contain pure profiles, \cite{vlatakis2020no}]\label{attractors contain nodes}
Every attractor of the replicator flow $\phi$ contains an attracting set of nodes in the response graph.
\end{lem}
Lemma~\ref{attractors contain nodes} establishes a connection between attractors and the response graph. In particular, because sink connected components are minimal attracting sets, it follows that every attractor of the replicator contains a sink connected component.

\subsection{Chains}

\begin{defn}[\cite{alongi2007recurrence}]
Let $\phi$ be the flow of a dynamical system on a compact metric space $X$. Let $x$ and $y$ be points in $X$. There is an \emph{$(\epsilon,T)$-chain} from $x$ to $y$ if there is a finite sequence of points $x_1,x_2,\dots,x_n$ with $x=x_1$ and $y=x_n$, and times $t_1,\dots,t_n \in [T,\infty)$ such that $\dist(\phi(x_i, t_i), x_{i+1}) < \epsilon$.
\end{defn}

\begin{defn}
If there is an $(\epsilon,T)$-chain from $x$ to $y$ for \emph{all} $\epsilon > 0$ and $T>0$ we say there is a \emph{pseudo-orbit} from $x$ to $y$, and write $\chain{x}{y}$. We say $x$ and $y$ are \emph{chain equivalent} if $\chain{x}{y}$ and $\chain{y}{x}$ and write $\chaineq{x}{y}$. If $\chaineq{x}{x}$ we say $x$ is \emph{chain recurrent}.
\end{defn}
The relation defined by pseudo-orbits is transitive, though not reflexive, as in general not all points are chain recurrent. However, restricting our attention to the chain recurrent points gives a preorder. Every preorder has an associated partial order given by grouping points into equivalence classes. Each equivalence class under chain equivalence is called a \emph{chain component}, and they are connected components of the topological space~\citep{alongi2007recurrence}. Given a chain recurrent point $x$, we denote its chain component by $[x]$. 
The chain recurrent points can be characterised by attractors.
\begin{thm}[\cite{alongi2007recurrence}] \label{recurrence characterisation}
A point $x$ is chain recurrent if and only if, for every attractor $A$, either $x\in A$ or $x\in A^*$.
\end{thm}
We write $\mathcal{A}_Y$ for the set of attractors of $\phi$ wholly containing a set $Y$ of points. If $Y$ is a singleton $\{y\}$ then we simplify notation by writing $\mathcal{A}_y$ rather than $\mathcal{A}_{\{y\}}$. Likewise, we write $\mathcal{R}_Y$ for the set of repellors containing $Y$. The following lemma is critical.
\begin{lem} \label{chains in attractors}
If $\chain{x}{y}$ then $\mathcal{A}_x\subseteq \mathcal{A}_y$ and $\mathcal{R}_y\subseteq \mathcal{R}_x$. Conversely, if $x$ is chain recurrent, then $\mathcal{A}_x \subseteq \mathcal{A}_y$ implies $\chain{x}{y}$, and $\mathcal{R}_x\subseteq \mathcal{R}_y$ implies $\chain{y}{x}$. Consequently, if $x$ is chain recurrent, $\chaineq{x}{y}$ if and only if $\mathcal{A}_x = \mathcal{A}_y$ and $\mathcal{R}_x = \mathcal{R}_y$. If $y$ is also chain recurrent then one of these equalities is sufficient.
\end{lem}
This states that pseudo-orbits never leave (but possibly enter) attractors. Dually, a pseudo-orbit never enters repellors, but may leave them. The following result connects the response graph with pseudo-orbits.
\begin{lem}[\cite{omidshafiei2019alpha}] \label{graph reachable}
If $v$ and $w$ are pure profiles, and there is a path $v,v_1,v_2,\dots,w$ in the response graph, then $\chain{v}{w}$ under the replicator.
\end{lem}
The idea of this proof is that the points on an arc move toward in the direction of the arc under the replicator. Upon nearing the head of the arc, an $\epsilon$-jump will take us to the subsequent arc of the path. It follows that if $H$ is a connected component of the response graph, then all profiles in $H$ are contained in exactly one chain component, which we denote $[H]$.

\section{The Existence Theorem for Sink Chain Components}

In this section we establish that sink chain components exist under the replicator, as conjectured by~\citeauthor{papadimitriou2019game}. Specifically we prove that all chain recurrent points in the game can reach a sink chain component via a pseudo-orbit. Combining this with the Fundamental Theorem of Dynamical Systems, which states that all points converge to chain recurrent ones, we find that under the replicator there are pseudo-orbits from all points to sink chain components.
\begin{thm}[Existence of sink chain components] \label{scc theorem}
Let $x\in X$ be a chain recurrent point under the replicator. Then there exists a pure profile $y$ and pseudo-orbit $\chain{x}{y}$ where $y$ is in a sink connected component $H$ contained in a sink chain component $[H]$.
\end{thm}

The proof proceeds in two steps. The first step is to show that there are pseudo-orbits from all chain recurrent points to a pure profile. Lemmas~\ref{attractors contain nodes} and \ref{graph reachable} establish that attractors of the replicator contain sink connected components. Lemma~\ref{chains in attractors} demonstrates that the chain components are defined by the structure of the attractors. If we knew that there were only finitely many attractors then we would be done---any minimal attractor is a sink chain component. To establish existence without assuming there are finitely many attractors we use the relationship with the sink connected components, which we know are finite in number. The key observation is that from any chain recurrent point $x$, there must exist a pure profile $p$ contained in all attractors containing $x$, which implies the existence of a pseudo-orbit from $x$ to $p$. Finally, we use the attractors of the system to order the sink connected components, and so find the sink chain components.

\section{Bounding the Sink Chain Components using Connected Components}

Theorem~\ref{scc theorem} establishes that sink chain components exist and contain sink connected components. Our goal in this section is to strengthen these results. That is, given a sink connected component $H$, what can we say about the unique sink chain component $[H]$ which contains it\footnote{Uniqueness here follows from the fact that chain components are disjoint. Each connected component $H$ is chain recurrent (Lemma~\ref{graph reachable}), and so is contained in a chain component.}?

We first prove a lower bound: $[H]$ always contains the \emph{content} of $H$.

\begin{defn}[Content] \label{def: content}
Let $W$ be a collection of pure profiles in a game $u$. The \emph{content} of $W$, denoted $\content(W)$, is the set of all mixed strategy profiles $x$ where all pure profiles in the support of $x$ are in $W$.
\end{defn}

The content is a topological subspace of the strategy space. If $W$ is a connected component of the response graph, then the content of $W$ is topologically connected. We now show that if $W$ is strongly connected in the graph, the content of $W$ is always a subset of $[W]$.

\begin{thm} \label{content containment}
Let $H$ be a strongly connected component of the response graph. Then the content of $H$ is contained in $[H]$, the chain component of the replicator containing $H$.
\end{thm}

The proof uses induction on subgames, following Theorem~\ref{subgame invariance}. Specifically, if the entire boundary of a subgame is contained in the same chain component, then there can be no attractors in that subgame, by Lemma~\ref{chains in attractors} and Theorem~\ref{no interior attractors}. As a corollary, whenever the response graph is strongly connected the \emph{entire strategy space is chain recurrent}. An example is the Matching Pennies game (Figure~\ref{fig: 2x2}).
\begin{corol} \label{corol: strong connectedness}
In any game with a strongly connected response graph, every mixed profile is chain recurrent.
\end{corol}

Now we can obtain an upper bound in some cases: when the sink connected component $H$ is a subgame, then $[H] = \content(H)$. This is because attracting subgames are attractors.

\begin{lem} \label{attracting subgames}
If $Y$ is an attracting subgame, then the set of mixed profiles in $Y$ is an attractor under the replicator.
\end{lem}
Pure NEs are attracting \emph{singleton subgames}---by this lemma, they are always attractors under the replicator, as we know from existing results~\citep{vlatakis2020no}. Combining this result with Theorem~\ref{content containment} gives us:
\begin{corol} \label{sink subgames}
If a sink component $H$ is a subgame, then $[H] = \content(H)$.
\end{corol}
\begin{proof}
Sink connected components are attracting, so $H$ is an attracting subgame, and thus is an attractor under the replicator (Lemma~\ref{attracting subgames}). The content of $H$ is exactly the mixed profiles in that subgame, and by Theorem~\ref{content containment}, all points in the content of $H$ are in $[H]$, and so are chain equivalent. However, as an attractor, no pseudo-orbit can leave this subgame by Lemma~\ref{chains in attractors} and so $[H] = \content(H)$.
\end{proof}

The sink connected component is generally not a subgame, but when this is the case our upper and lower bounds meet, and the sink chain component can be precisely characterised.

\subsection{Zero-sum and Potential games}

A number of special properties are known for potential~\citep{monderer1996potential} and zero-sum~\citep{von2007theory} games. The current understanding of the chain components of such games comes from~\cite{piliouras2014optimization,papadimitriou2016nash,papadimitriou2018nash,mertikopoulos2018cycles} where it is proved that the chain components of a weighted potential game are its pure Nash equilibria, and when a two-player zero-sum game has an interior equilibrium, all profiles are chain recurrent. We generalise the former result in Theorem~\ref{sinks in preference potential} as a straightforward consequence of Corollary~\ref{sink subgames}.

The case of two-player zero-sum games is more intriguing. All two-player zero-sum games with interior equilibria that we know of have \emph{strongly connected response graphs}. That is, in all known cases, this result is subsumed by Corollary~\ref{corol: strong connectedness}. Indeed, if Conjecture~\ref{sink ccs = content sink eq} is true, then this would become a theorem: any two-player zero-sum game with an interior equilibrium would have a strongly connected response graph
. Corollary~\ref{corol: strong connectedness} is much broader. Many games are strongly connected but not two-player zero-sum with an interior equilibrium. Focusing on two-player games, we find that the response graphs of all $2\times n$ zero-sum games are strongly connected (Lemma~\ref{2xn connected}). On $3\times 3$ strict games, there are only two response graphs of zero-sum games which are \emph{not} strongly connected~\citep{biggar2022twoplayer}. These are the Inner and Outer Diamond graphs~\citep{biggar2022twoplayer}, depicted in Figure~\ref{fig: inner and outer diamond}. No zero-sum game with either of these response graphs ever has an interior equilibrium (Lemma~\ref{lem: diamond games}). 

Following~\cite{biggar2022twoplayer}, we call a game \emph{preference-zero-sum} if its response graph is isomorphic to that of a two-player zero-sum game. We call it \emph{preference-potential} if its response graph is isomorphic to that of a potential game. This equivalence relation is quite broad. In particular, any weighted potential game is preference-potential, but not vice versa.

\begin{thm} \label{sinks in preference potential}
In a strict preference-potential game under the replicator, the sink chain components are exactly $\{p_1\},\dots,\{p_m\}$ where $p_1,\dots,p_m$ are the pure Nash equilibria.
\end{thm}
\begin{proof}
Each pure Nash equilibrium $p$ in a strict game is an attracting subgame, so by Corollary~\ref{sink subgames} the associated sink chain component is precisely $\{p\}$. Since every sink chain component contains a sink connected component, and thus a pure Nash equilibrium, these are precisely the sink chain components.
\end{proof}

Recalling the Coordination game (Figure~\ref{fig: 2x2}), preference-potential games can have mixed Nash equilibria, but these are \emph{never in the sink chain components}, while the pure Nash equilibria are always precisely the sink chain components. Beyond preference-potential games, this argument applies to any game where every sink connected component is a subgame, such as \emph{weakly acyclic games}~\citep{young1993evolution}.

\begin{lem} \label{lem: unique}
Every preference-zero-sum game has exactly one sink connected component and thus precisely one sink chain component.
\end{lem}
\begin{proof}
A sink chain component exists by Theorem~\ref{scc theorem}. Every preference-zero-sum game has exactly one sink connected component by \cite[Theorem 4.9]{biggar2022twoplayer}. Sink chain components are disjoint and contain sink connected components, so there must be precisely one.
\end{proof}

We conclude that, at least in preference-zero-sum and preference-potential games, there is a one-to-one correspondence between sink chain components and sink connected components.

\section{Applications} \label{sec: applications}

\begin{figure}
    \centering
    \def\pAA{1}
\def\pAB{1}
\def\pBA{1}
\def\pBB{1}

\newcommand{\replicatorU}[2]{#1 * (1.0 - #1) * (\pAA * #2 + \pAB * (1.0 - #2))}
\newcommand{\replicatorV}[2]{#2 * (1.0 - #2) * (\pBA * #1 + \pBB * (1.0 - #1))}
\newcommand{\replicatorlength}[2]{sqrt{(\replicatorU{#1}{#2})^2+(\replicatorV{#1}{#2})^2}}
\begin{tikzpicture}
\begin{axis}[
    xmin=-0.15, xmax=1.15,
    ymin=-0.15, ymax=1.15,
    zmin = 0, zmax = 1,
    axis equal image,
    ticks = none,
    view = {0}{90},
    scale = 1,
    height=4cm,
    axis line style = {opacity=0},
    xtick style = {draw=none},
    ytick style = {draw=none},
]
\path[fill = gray!30] (1,1,0) circle (6pt);
\begin{scope}[very thick,decoration={
    markings,
    mark=at position 0.62 with {\arrow{>}}}
    ] 
\node[draw,circle,color=black!50,inner sep=0pt] (000) at (0,0,0) {};
\node[draw,circle,color=black!50,inner sep=0pt] (010) at (0,1,0) {};
\node[draw,circle,color=black!50,inner sep=0pt] (100) at (1,0,0) {};
\node[draw,circle,color=black!50,inner sep=0pt] (110) at (1,1,0) {};
\draw[draw,color=black!50,line width=1pt,postaction={decorate}] (000) to (010);
\draw[draw,color=black!50,line width=1pt,postaction={decorate}] (000) to (100);
\draw[draw,color=black!50,line width=1pt,postaction={decorate}] (010) to (110);
\draw[draw,color=black!50,line width=1pt,postaction={decorate}] (100) to (110);
\end{scope}
 
\addplot3[
    point meta = {\replicatorlength{x}{y}},
    quiver = {
        u = {\replicatorU{x}{y}},
        v = {\replicatorV{x}{y}},
        scale arrows = 0.4,
    },
    -stealth,
    domain = 0.1:0.9,
    domain y = 0.1:0.9,
    samples = 6,
] {0};
 
\end{axis}
 
\end{tikzpicture}
 
    \def\pAA{1}
\def\pAB{1}
\def\pBA{-1}
\def\pBB{1}

\newcommand{\replicatorU}[2]{#1 * (1.0 - #1) * (\pAA * #2 + \pAB * (1.0 - #2))}
\newcommand{\replicatorV}[2]{#2 * (1.0 - #2) * (\pBA * #1 + \pBB * (1.0 - #1))}
\newcommand{\replicatorlength}[2]{sqrt{(\replicatorU{#1}{#2})^2+(\replicatorV{#1}{#2})^2}}
\begin{tikzpicture}
\begin{axis}[
    xmin=-0.15, xmax=1.15,
    ymin=-0.15, ymax=1.15,
    zmin = 0, zmax = 1,
    axis equal image,
    ticks = none,
    view = {0}{90},
    scale = 1,
    height=4cm,
    axis line style = {opacity=0},
    xtick style = {draw=none},
    ytick style = {draw=none},
]
\path[fill = gray!30] (1,0,0) circle (6pt);
\begin{scope}[very thick,decoration={
    markings,
    mark=at position 0.62 with {\arrow{>}}}
    ] 
\node[draw,circle,color=black!50,inner sep=0pt] (000) at (0,0,0) {};
\node[draw,circle,color=black!50,inner sep=0pt] (010) at (0,1,0) {};
\node[draw,circle,color=black!50,inner sep=0pt] (100) at (1,0,0) {};
\node[draw,circle,color=black!50,inner sep=0pt] (110) at (1,1,0) {};
\draw[draw,color=black!50,line width=1pt,postaction={decorate}] (000) to (010);
\draw[draw,color=black!50,line width=1pt,postaction={decorate}] (000) to (100);
\draw[draw,color=black!50,line width=1pt,postaction={decorate}] (010) to (110);
\draw[draw,color=black!50,line width=1pt,postaction={decorate}] (110) to (100);
\end{scope}
 
\addplot3[
    point meta = {\replicatorlength{x}{y}},
    quiver = {
        u = {\replicatorU{x}{y}},
        v = {\replicatorV{x}{y}},
        scale arrows = 0.4,
    },
    -stealth,
    domain = 0.1:0.9,
    domain y = 0.1:0.9,
    samples = 6,
] {0};
 
\end{axis}
 
\end{tikzpicture}
 
    \def\pAA{-1}
\def\pAB{1}
\def\pBA{-1}
\def\pBB{1}

\newcommand{\replicatorU}[2]{#1 * (1.0 - #1) * (\pAA * #2 + \pAB * (1.0 - #2))}
\newcommand{\replicatorV}[2]{#2 * (1.0 - #2) * (\pBA * #1 + \pBB * (1.0 - #1))}
\newcommand{\replicatorlength}[2]{sqrt{(\replicatorU{#1}{#2})^2+(\replicatorV{#1}{#2})^2}}
\begin{tikzpicture}
\begin{axis}[
    xmin=-0.15, xmax=1.15,
    ymin=-0.15, ymax=1.15,
    zmin = 0, zmax = 1,
    axis equal image,
    ticks = none,
    view = {0}{90},
    scale = 1,
    height=4cm,
    axis line style = {opacity=0},
    xtick style = {draw=none},
    ytick style = {draw=none},
]
\path[fill = gray!30] (1,0,0) circle (6pt);
\path[fill = gray!30] (0,1,0) circle (6pt);
\begin{scope}[very thick,decoration={
    markings,
    mark=at position 0.62 with {\arrow{>}}}
    ] 
\node[draw,circle,color=black!50,inner sep=0pt] (000) at (0,0,0) {};
\node[draw,circle,color=black!50,inner sep=0pt] (010) at (0,1,0) {};
\node[draw,circle,color=black!50,inner sep=0pt] (100) at (1,0,0) {};
\node[draw,circle,color=black!50,inner sep=0pt] (110) at (1,1,0) {};
\draw[draw,color=black!50,line width=1pt,postaction={decorate}] (000) to (010);
\draw[draw,color=black!50,line width=1pt,postaction={decorate}] (000) to (100);
\draw[draw,color=black!50,line width=1pt,postaction={decorate}] (110) to (010);
\draw[draw,color=black!50,line width=1pt,postaction={decorate}] (110) to (100);
\end{scope}
 
\addplot3[
    point meta = {\replicatorlength{x}{y}},
    quiver = {
        u = {\replicatorU{x}{y}},
        v = {\replicatorV{x}{y}},
        scale arrows = 0.8,
    },
    -stealth,
    domain = 0.1:0.9,
    domain y = 0.1:0.9,
    samples = 6,
] {0};
 
\end{axis}
 
\end{tikzpicture}
 
    \def\pAA{1}
\def\pAB{-1}
\def\pBA{-1}
\def\pBB{1}

\newcommand{\replicatorU}[2]{#1 * (1.0 - #1) * (\pAA * #2 + \pAB * (1.0 - #2))}
\newcommand{\replicatorV}[2]{#2 * (1.0 - #2) * (\pBA * #1 + \pBB * (1.0 - #1))}
\newcommand{\replicatorlength}[2]{sqrt{(\replicatorU{#1}{#2})^2+(\replicatorV{#1}{#2})^2}}
\begin{tikzpicture}
\begin{axis}[
    xmin=-0.15, xmax=1.15,
    ymin=-0.15, ymax=1.15,
    zmin = 0, zmax = 1,
    axis equal image,
    ticks = none,
    view = {0}{90},
    scale = 1,
    height=4cm,
    axis line style = {opacity=0},
    xtick style = {draw=none},
    ytick style = {draw=none},
]
\path [fill = gray!30 ] (0,0,0) -- (0,1,0) -- (1,1,0) -- (1,0,0);
\begin{scope}[very thick,decoration={
    markings,
    mark=at position 0.56 with {\arrow{>}}}
    ] 
\node[draw,circle,color=black!50,inner sep=0pt] (000) at (0,0,0) {};
\node[draw,circle,color=black!50,inner sep=0pt] (010) at (0,1,0) {};
\node[draw,circle,color=black!50,inner sep=0pt] (100) at (1,0,0) {};
\node[draw,circle,color=black!50,inner sep=0pt] (110) at (1,1,0) {};
\draw[draw,color=black!50,line width=1pt,postaction={decorate}] (000) to (010);
\draw[draw,color=black!50,line width=1pt,postaction={decorate}] (100) to (000);
\draw[draw,color=black!50,line width=1pt,postaction={decorate}] (010) to (110);
\draw[draw,color=black!50,line width=1pt,postaction={decorate}] (110) to (100);
\end{scope}
 
\addplot3[
    point meta = {\replicatorlength{x}{y}},
    quiver = {
        u = {\replicatorU{x}{y}},
        v = {\replicatorV{x}{y}},
        scale arrows = 0.7,
    },
    -stealth,
    domain = 0.1:0.9,
    domain y = 0.1:0.9,
    samples = 6,
] {0};
 
\end{axis}
 
\end{tikzpicture}
 
    \caption{Sink chain components (highlighted in grey) of the replicator dynamic on the strict 2x2 games. These response graphs are called Double-Dominance, Single-Dominance, Coordination and Matching Pennies~\citep{biggar2022twoplayer}. We have overlaid a typical example of the replicator vector field for clarity in each case.}
    \label{fig: 2x2}
\end{figure}
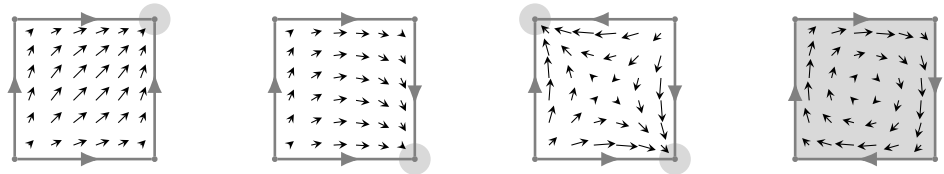

In this section we identify the sink chain components of the replicator on a collection of small games, demonstrating how our findings extend the existing literature.

\begin{itemize}
    \item \textbf{(Any game equivalent to) $2\times 2$ Coordination}: This game has been well-studied~\citep{papadimitriou2016nash,papadimitriou2018nash,papadimitriou2019game}. It is known that any game sharing a response graph with $2\times 2$ Coordination (Figure~\ref{fig: 2x2}) has two sink chain components which are the two pure Nash equilibria. All such games are preference-potential, and so this is an immediate corollary of Theorem~\ref{sinks in preference potential}.
    \item \textbf{(Any game equivalent to) Matching Pennies}: Again, this game is well-studied~\citep{papadimitriou2016nash,papadimitriou2018nash,papadimitriou2019game,balduzzi2018mechanics}, and it is known that all profiles are chain recurrent. However, it is assumed in these papers that the game is zero-sum. Because the response graph of Matching Pennies is a cycle, it is strongly connected and so by Corollary~\ref{corol: strong connectedness} any game sharing a response graph with Matching Pennies is chain recurrent. Almost all such games are not zero-sum.
    \item \textbf{\emph{Any} $2\times 2$ strict game}: There are only four response graphs of strict $2\times 2$ games~\citep{biggar2022twoplayer}, of which we have already discussed two (Matching Pennies and $2\times 2$ Coordination). The remaining two (Single- and Double-dominance) are preference-potential, so Theorem~\ref{sinks in preference potential} applies. See Figure~\ref{fig: 2x2}.
    \item \textbf{\emph{Any} $2\times 3$ strict game}: We can focus on games without dominated strategies, because the sink chain component will never contain such strategies~\citep{sandholm2010population}. By~\cite{biggar2022twoplayer}, there are three response graphs of $2\times 3$ games without dominated strategies. One is strongly connected, and so all profiles are chain recurrent (Corollary~\ref{corol: strong connectedness}); another is preference-potential; in the final graph the unique sink connected component is a pure Nash equilibrium, so in both of the latter cases the sink chain components are exactly the pure Nash equilibria (Theorem~\ref{sinks in preference potential}).
    \item \textbf{\emph{Any} $2\times n$ strict game}: Generalising the previous cases, it turns out (Lemma~\ref{2xn connected}) that sink connected components of strict $2\times n$ games are always subgames, so Corollary~\ref{sink subgames} applies. Additionally, every preference-zero-sum $2\times n$ game without dominated strategies is strongly connected.
    \item \textbf{\emph{Any} $3\times 3$ strict preference-zero-sum game}: By~\cite{biggar2022twoplayer}, all but two of the response graphs of such games are strongly connected, so Corollary~\ref{corol: strong connectedness} applies. The two exceptions are known as the Inner and Outer Diamond graphs (Figure~\ref{fig: inner and outer diamond}). The unique (Lemma~\ref{lem: unique}) sink connected component in the Inner Diamond game is a pure Nash equilibrium, so this is the unique sink chain component by Theorem~\ref{sinks in preference potential}. The Outer Diamond game is also preference-zero-sum, so has a unique sink chain component containing the content of its sink connected component. We conjecture (Conjecture~\ref{sink ccs = content sink eq}) that the sink chain component is always precisely the content of the sink connected component.
\end{itemize}

\begin{figure}
    \centering
    \begin{tikzpicture}

    \node[inner sep=1pt] (1) at (0,0) {};
    \node[inner sep=1pt] (2) at (1,0) {};
    \node[inner sep=1pt] (3) at (2,0) {};
    \node[inner sep=1pt] (4) at (0,1) {};
    \node[inner sep=1pt] (5) at (1,1) {};
    \node[inner sep=1pt] (6) at (2,1) {};
    \node[inner sep=1pt] (7) at (0,2) {};
    \node[inner sep=1pt] (8) at (1,2) {};
    \node[inner sep=1pt] (9) at (2,2) {};
    
    \path[fill = gray!30] (0,2) circle (8pt);
    
    \draw[->] (1) to (4);
    \draw[->] (1) to[out = -20, in = 200] (3);
    \draw[->] (3) to[out = 70, in = -70] (9);
    \draw[->] (9) to (8);
    \draw[->] (5) to (8);
    \draw[->] (4) to (5);
    
    \draw[->] (4) to (7);
    \draw[->] (1) to[out = 110, in = -110] (7);
    \draw[->] (8) to (7);
    \draw[->] (9) to[out = 160, in = 20] (7);
    \draw[->] (2) to (1);
    \draw[->] (2) to (3);
    \draw[->] (5) to (2);
    \draw[->] (8) to[out = -70, in = 70] (2);
    \draw[->] (3) to (6);
    \draw[->] (9) to (6);
    \draw[->] (6) to (5);
    \draw[->] (6) to[out = 200, in = -20] (4);
    

 \end{tikzpicture}
    \qquad
    \begin{tikzpicture}

    \node[inner sep=1pt] (1) at (0,0) {};
    \node[inner sep=1pt] (2) at (1,0) {};
    \node[inner sep=1pt] (3) at (2,0) {};
    \node[inner sep=1pt] (4) at (0,1) {};
    \node[inner sep=1pt] (5) at (1,1) {};
    \node[inner sep=1pt] (6) at (2,1) {};
    \node[inner sep=1pt] (7) at (0,2) {};
    \node[inner sep=1pt] (8) at (1,2) {};
    \node[inner sep=1pt] (9) at (2,2) {};
    
    \path [fill = gray!30] (0,0) -- (0,1) -- (1,1) -- (1,2) -- (2,2) to[out = -70, in = 70] (2,0) to[in = -20, out = 200] (0,0);
    
    \draw[->] (4) to (1);
    \draw[->] (3) to[in = -20, out = 200] (1);
    \draw[->] (9) to[in = 70, out = -70] (3);
    \draw[->] (8) to (9);
    \draw[->] (8) to (5);
    \draw[->] (5) to (4);
    
    \draw[->] (7) to (4);
    \draw[->] (7) to[in = 110, out = -110] (1);
    \draw[->] (7) to (8);
    \draw[->] (7) to[in = 160, out = 20] (9);
    \draw[->] (1) to (2);
    \draw[->] (3) to (2);
    \draw[->] (2) to (5);
    \draw[->] (2) to[in = -70, out = 70] (8);
    \draw[->] (6) to (3);
    \draw[->] (6) to (9);
    \draw[->] (5) to (6);
    \draw[->] (4) to[in = 200, out = -20] (6);
    

 \end{tikzpicture}
    \caption{The Inner and Outer Diamond graphs and the content of their sink connected components.}
    \label{fig: inner and outer diamond}
\end{figure}
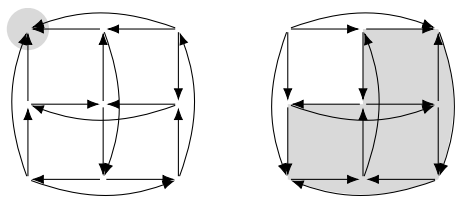

\section{Conclusions and Open Problems} \label{sec: conjectures}

In this paper we demonstrated the influence of the response graph on the chain components of the replicator dynamic. In particular, we showed that sink chain components exist and contain the content of sink connected components of the response graph.

The close relationship between connected components and chain components is interesting from the perspective of both game theory and dynamical systems, and poses several questions that are deserving of further investigation. In this section we outline some of the most important. In all the cases we discussed above where we could identify the sink chain components, each not only contained a sink connected component but in fact exactly one, and each sink connected component was contained in a sink chain component. We conjecture that this relationship is generally true.
\begin{conj} \label{conj: one to one}
There is a one-to-one correspondence between sink chain components and sink connected components.
\end{conj}

Using the construction in the proof of Theorem~\ref{scc theorem}, we can cast this question as one about attractors: Conjecture~\ref{conj: one to one} is true if and only if each sink connected component has at least one attractor containing it and no others. If true, it would greatly strengthen the premise of~\citeauthor{papadimitriou2019game} that sink connected components are the \emph{right} combinatorial surrogate of sink chain components. However, even if such a one-to-one correspondence existed, it wouldn't necessarily specify precisely which points are in sink chain components. The following stronger conjecture does do this:

\begin{conj} \label{sink ccs = content sink eq}
The sink chain components of a game are exactly the content of the sink connected components.
\end{conj}

Again, we can rephrase this as a statement about attractors: the conjecture holds if and only if the content of a sink chain component is always an attractor. Conjecture~\ref{sink ccs = content sink eq} would be particularly shocking, because it would imply that the `long-run' outcome of the replicator dynamic is \emph{entirely dictated} by the response graph. That is, from the perspective of the replicator, only the preference orders of each player affect the chain components. We do not currently know of counterexamples to either conjecture. 

There is another important question: to which other dynamics do these results generalise? In particular, the steep FTRL dynamics~\citep{vlatakis2020no} are known to satisfy Theorem~\ref{no interior attractors}, which is necessary for these results, so it seems plausible that these results may apply to this class.

\bibliography{references}

\begin{thebibliography}{43}
\providecommand{\natexlab}[1]{#1}
\providecommand{\url}[1]{\texttt{#1}}
\expandafter\ifx\csname urlstyle\endcsname\relax
  \providecommand{\doi}[1]{doi: #1}\else
  \providecommand{\doi}{doi: \begingroup \urlstyle{rm}\Url}\fi

\bibitem[Akin and Losert(1984)]{akin1984evolutionary}
Ethan Akin and Viktor Losert.
\newblock Evolutionary dynamics of zero-sum games.
\newblock \emph{Journal of Mathematical Biology}, 20\penalty0 (3):\penalty0
  231--258, 1984.

\bibitem[Alongi and Nelson(2007)]{alongi2007recurrence}
John~M Alongi and Gail~Susan Nelson.
\newblock \emph{Recurrence and topology}, volume~85.
\newblock American Mathematical Soc., 2007.

\bibitem[Andrade et~al.(2021)Andrade, Frongillo, and
  Piliouras]{andrade2021learning}
Gabriel~P Andrade, Rafael Frongillo, and Georgios Piliouras.
\newblock Learning in matrix games can be arbitrarily complex.
\newblock In \emph{Conference on Learning Theory}, pages 159--185. PMLR, 2021.

\bibitem[Arora et~al.(2012)Arora, Hazan, and Kale]{arora2012multiplicative}
Sanjeev Arora, Elad Hazan, and Satyen Kale.
\newblock The multiplicative weights update method: a meta-algorithm and
  applications.
\newblock \emph{Theory of computing}, 8\penalty0 (1):\penalty0 121--164, 2012.

\bibitem[Balduzzi et~al.(2018)Balduzzi, Racaniere, Martens, Foerster, Tuyls,
  and Graepel]{balduzzi2018mechanics}
David Balduzzi, Sebastien Racaniere, James Martens, Jakob Foerster, Karl Tuyls,
  and Thore Graepel.
\newblock The mechanics of n-player differentiable games.
\newblock In \emph{International Conference on Machine Learning}, pages
  354--363. PMLR, 2018.

\bibitem[Bena{\"\i}m et~al.(2012)Bena{\"\i}m, Hofbauer, and
  Sorin]{benaim2012perturbations}
Michel Bena{\"\i}m, Josef Hofbauer, and Sylvain Sorin.
\newblock Perturbations of set-valued dynamical systems, with applications to
  game theory.
\newblock \emph{Dynamic Games and Applications}, 2\penalty0 (2):\penalty0
  195--205, 2012.

\bibitem[Biggar and Shames(2022)]{biggar2022twoplayer}
Oliver Biggar and Iman Shames.
\newblock The graph structure of two-player games.
\newblock \emph{arXiv preprint arXiv:2209.10182}, 2022.

\bibitem[Candogan et~al.(2011)Candogan, Menache, Ozdaglar, and
  Parrilo]{candogan2011flows}
Ozan Candogan, Ishai Menache, Asuman Ozdaglar, and Pablo~A Parrilo.
\newblock Flows and decompositions of games: Harmonic and potential games.
\newblock \emph{Mathematics of Operations Research}, 36\penalty0 (3):\penalty0
  474--503, 2011.

\bibitem[Cheung and Piliouras(2019)]{cheung2019vortices}
Yun~Kuen Cheung and Georgios Piliouras.
\newblock Vortices instead of equilibria in minmax optimization: Chaos and
  butterfly effects of online learning in zero-sum games.
\newblock In \emph{Conference on Learning Theory}, pages 807--834. PMLR, 2019.

\bibitem[Cheung and Piliouras(2020)]{cheung2020chaos}
Yun~Kuen Cheung and Georgios Piliouras.
\newblock Chaos, extremism and optimism: Volume analysis of learning in games.
\newblock \emph{Advances in Neural Information Processing Systems},
  33:\penalty0 9039--9049, 2020.

\bibitem[Conley(1978)]{conley1978isolated}
Charles~C Conley.
\newblock \emph{Isolated invariant sets and the Morse index}.
\newblock Number~38. American Mathematical Soc., 1978.

\bibitem[Daskalakis et~al.(2009)Daskalakis, Goldberg, and
  Papadimitriou]{daskalakis2009complexity}
Constantinos Daskalakis, Paul~W Goldberg, and Christos~H Papadimitriou.
\newblock The complexity of computing a nash equilibrium.
\newblock \emph{SIAM Journal on Computing}, 39\penalty0 (1):\penalty0 195--259,
  2009.

\bibitem[Eshel et~al.(1983)Eshel, Akin, et~al.]{eshel1983coevolutionary}
Ilan Eshel, Ethan Akin, et~al.
\newblock Coevolutionary instability of mixed nash solutions.
\newblock \emph{Journal of Mathematical Biology}, 18\penalty0 (2):\penalty0
  123--133, 1983.

\bibitem[Fabrikant and Papadimitriou(2008)]{fabrikant2008complexity}
Alex Fabrikant and Christos~H Papadimitriou.
\newblock The complexity of game dynamics: Bgp oscillations, sink equilibria,
  and beyond.
\newblock In \emph{SODA}, volume~8, pages 844--853. Citeseer, 2008.

\bibitem[Goemans et~al.(2005)Goemans, Mirrokni, and Vetta]{goemans2005sink}
Michel Goemans, Vahab Mirrokni, and Adrian Vetta.
\newblock Sink equilibria and convergence.
\newblock In \emph{46th Annual IEEE Symposium on Foundations of Computer
  Science (FOCS'05)}, pages 142--151. IEEE, 2005.

\bibitem[Hart and Mas-Colell(2003)]{hart2003uncoupled}
Sergiu Hart and Andreu Mas-Colell.
\newblock Uncoupled dynamics do not lead to nash equilibrium.
\newblock \emph{American Economic Review}, 93\penalty0 (5):\penalty0
  1830--1836, 2003.

\bibitem[Hofbauer(1996)]{hofbauer1996evolutionary}
Josef Hofbauer.
\newblock Evolutionary dynamics for bimatrix games: A hamiltonian system?
\newblock \emph{Journal of Mathematical Biology}, 34\penalty0 (5):\penalty0
  675--688, 1996.

\bibitem[Hofbauer and Sigmund(2003)]{hofbauer2003evolutionary}
Josef Hofbauer and Karl Sigmund.
\newblock Evolutionary game dynamics.
\newblock \emph{Bulletin of the American mathematical society}, 40\penalty0
  (4):\penalty0 479--519, 2003.

\bibitem[Kalies et~al.(2021)Kalies, Mischaikow, and
  Vandervorst]{kalies2021lattice}
William~D Kalies, Konstantin Mischaikow, and Robert~CAM Vandervorst.
\newblock Lattice structures for attractors iii.
\newblock \emph{Journal of Dynamics and Differential Equations}, pages 1--40,
  2021.

\bibitem[Kleinberg et~al.(2011)Kleinberg, Ligett, Piliouras, and
  Tardos]{kleinberg2011beyond}
Robert~D Kleinberg, Katrina Ligett, Georgios Piliouras, and {\'E}va Tardos.
\newblock Beyond the nash equilibrium barrier.
\newblock In \emph{ICS}, pages 125--140, 2011.

\bibitem[Mertens(2004)]{mertens2004ordinality}
Jean-Fran{\c{c}}ois Mertens.
\newblock Ordinality in non cooperative games.
\newblock \emph{International Journal of Game Theory}, 32\penalty0
  (3):\penalty0 387--430, 2004.

\bibitem[Mertikopoulos et~al.(2018)Mertikopoulos, Papadimitriou, and
  Piliouras]{mertikopoulos2018cycles}
Panayotis Mertikopoulos, Christos Papadimitriou, and Georgios Piliouras.
\newblock Cycles in adversarial regularized learning.
\newblock In \emph{Proceedings of the Twenty-Ninth Annual ACM-SIAM Symposium on
  Discrete Algorithms}, pages 2703--2717. SIAM, 2018.

\bibitem[Milionis et~al.(2022)Milionis, Papadimitriou, Piliouras, and
  Spendlove]{milionis2022nash}
Jason Milionis, Christos Papadimitriou, Georgios Piliouras, and Kelly
  Spendlove.
\newblock Nash, conley, and computation: Impossibility and incompleteness in
  game dynamics.
\newblock \emph{arXiv preprint arXiv:2203.14129}, 2022.

\bibitem[Monderer and Shapley(1996)]{monderer1996potential}
Dov Monderer and Lloyd~S Shapley.
\newblock Potential games.
\newblock \emph{Games and economic behavior}, 14\penalty0 (1):\penalty0
  124--143, 1996.

\bibitem[Nash(1951)]{nash1951non}
John Nash.
\newblock Non-cooperative games.
\newblock \emph{Annals of mathematics}, pages 286--295, 1951.

\bibitem[Omidshafiei et~al.(2019)Omidshafiei, Papadimitriou, Piliouras, Tuyls,
  Rowland, Lespiau, Czarnecki, Lanctot, Perolat, and
  Munos]{omidshafiei2019alpha}
Shayegan Omidshafiei, Christos Papadimitriou, Georgios Piliouras, Karl Tuyls,
  Mark Rowland, Jean-Baptiste Lespiau, Wojciech~M Czarnecki, Marc Lanctot,
  Julien Perolat, and Remi Munos.
\newblock $\alpha$-rank: Multi-agent evaluation by evolution.
\newblock \emph{Scientific reports}, 9\penalty0 (1):\penalty0 1--29, 2019.

\bibitem[Omidshafiei et~al.(2020)Omidshafiei, Tuyls, Czarnecki, Santos,
  Rowland, Connor, Hennes, Muller, P{\'e}rolat, Vylder,
  et~al.]{omidshafiei2020navigating}
Shayegan Omidshafiei, Karl Tuyls, Wojciech~M Czarnecki, Francisco~C Santos,
  Mark Rowland, Jerome Connor, Daniel Hennes, Paul Muller, Julien P{\'e}rolat,
  Bart~De Vylder, et~al.
\newblock Navigating the landscape of multiplayer games.
\newblock \emph{Nature communications}, 11\penalty0 (1):\penalty0 1--17, 2020.

\bibitem[Papadimitriou and Piliouras(2016)]{papadimitriou2016nash}
Christos Papadimitriou and Georgios Piliouras.
\newblock From nash equilibria to chain recurrent sets: Solution concepts and
  topology.
\newblock In \emph{Proceedings of the 2016 ACM Conference on Innovations in
  Theoretical Computer Science}, pages 227--235, 2016.

\bibitem[Papadimitriou and Piliouras(2018)]{papadimitriou2018nash}
Christos Papadimitriou and Georgios Piliouras.
\newblock From nash equilibria to chain recurrent sets: An algorithmic solution
  concept for game theory.
\newblock \emph{Entropy}, 20\penalty0 (10):\penalty0 782, 2018.

\bibitem[Papadimitriou and Piliouras(2019)]{papadimitriou2019game}
Christos Papadimitriou and Georgios Piliouras.
\newblock Game dynamics as the meaning of a game.
\newblock \emph{ACM SIGecom Exchanges}, 16\penalty0 (2):\penalty0 53--63, 2019.

\bibitem[Piliouras and Shamma(2014)]{piliouras2014optimization}
Georgios Piliouras and Jeff~S Shamma.
\newblock Optimization despite chaos: Convex relaxations to complex limit sets
  via poincar{\'e} recurrence.
\newblock In \emph{Proceedings of the twenty-fifth annual ACM-SIAM Symposium on
  Discrete Algorithms}, pages 861--873. SIAM, 2014.

\bibitem[Roughgarden(2010)]{roughgarden2010algorithmic}
Tim Roughgarden.
\newblock Algorithmic game theory.
\newblock \emph{Communications of the ACM}, 53\penalty0 (7):\penalty0 78--86,
  2010.

\bibitem[Sandholm(2010)]{sandholm2010population}
William~H Sandholm.
\newblock \emph{Population games and evolutionary dynamics}.
\newblock MIT press, 2010.

\bibitem[Selten(1988)]{selten1988note}
Reinhard Selten.
\newblock A note on evolutionarily stable strategies in asymmetric animal
  conflicts.
\newblock In \emph{Models of Strategic Rationality}, pages 67--75. Springer,
  1988.

\bibitem[Silver et~al.(2016)Silver, Huang, Maddison, Guez, Sifre, Van
  Den~Driessche, Schrittwieser, Antonoglou, Panneershelvam, Lanctot,
  et~al.]{silver2016mastering}
David Silver, Aja Huang, Chris~J Maddison, Arthur Guez, Laurent Sifre, George
  Van Den~Driessche, Julian Schrittwieser, Ioannis Antonoglou, Veda
  Panneershelvam, Marc Lanctot, et~al.
\newblock Mastering the game of go with deep neural networks and tree search.
\newblock \emph{Nature}, 529\penalty0 (7587):\penalty0 484--489, 2016.

\bibitem[Silver et~al.(2017)Silver, Schrittwieser, Simonyan, Antonoglou, Huang,
  Guez, Hubert, Baker, Lai, Bolton, et~al.]{silver2017mastering}
David Silver, Julian Schrittwieser, Karen Simonyan, Ioannis Antonoglou, Aja
  Huang, Arthur Guez, Thomas Hubert, Lucas Baker, Matthew Lai, Adrian Bolton,
  et~al.
\newblock Mastering the game of go without human knowledge.
\newblock \emph{Nature}, 550\penalty0 (7676):\penalty0 354--359, 2017.

\bibitem[Silver et~al.(2018)Silver, Hubert, Schrittwieser, Antonoglou, Lai,
  Guez, Lanctot, Sifre, Kumaran, Graepel, et~al.]{silver2018general}
David Silver, Thomas Hubert, Julian Schrittwieser, Ioannis Antonoglou, Matthew
  Lai, Arthur Guez, Marc Lanctot, Laurent Sifre, Dharshan Kumaran, Thore
  Graepel, et~al.
\newblock A general reinforcement learning algorithm that masters chess, shogi,
  and go through self-play.
\newblock \emph{Science}, 362\penalty0 (6419):\penalty0 1140--1144, 2018.

\bibitem[Smith and Price(1973)]{smith1973logic}
JMPGR Smith and George~R Price.
\newblock The logic of animal conflict.
\newblock \emph{Nature}, 246\penalty0 (5427):\penalty0 15--18, 1973.

\bibitem[Strogatz(2018)]{strogatz2018nonlinear}
Steven~H Strogatz.
\newblock \emph{Nonlinear dynamics and chaos: with applications to physics,
  biology, chemistry, and engineering}.
\newblock CRC press, 2018.

\bibitem[Vlatakis-Gkaragkounis et~al.(2020)Vlatakis-Gkaragkounis, Flokas,
  Lianeas, Mertikopoulos, and Piliouras]{vlatakis2020no}
Emmanouil-Vasileios Vlatakis-Gkaragkounis, Lampros Flokas, Thanasis Lianeas,
  Panayotis Mertikopoulos, and Georgios Piliouras.
\newblock No-regret learning and mixed nash equilibria: They do not mix.
\newblock \emph{Advances in Neural Information Processing Systems},
  33:\penalty0 1380--1391, 2020.

\bibitem[Von~Neumann and Morgenstern(1944)]{von2007theory}
John Von~Neumann and Oskar Morgenstern.
\newblock \emph{Theory of games and economic behavior}.
\newblock Princeton university press, 1944.

\bibitem[Yang and Wang(2020)]{yang2020overview}
Yaodong Yang and Jun Wang.
\newblock An overview of multi-agent reinforcement learning from game
  theoretical perspective.
\newblock \emph{arXiv preprint arXiv:2011.00583}, 2020.

\bibitem[Young(1993)]{young1993evolution}
H~Peyton Young.
\newblock The evolution of conventions.
\newblock \emph{Econometrica: Journal of the Econometric Society}, pages
  57--84, 1993.

\end{thebibliography}

\appendix

\section{Proofs}

In addition to the notation introduced in Section~\ref{sec: preliminaries}, we will use some additional notation in the proofs in this appendix. If $Y$ is a subgame, we denote the strategies used by player $p$ in $Y$ by $Y_p$. Similarly, we denote the $p$-antiprofiles where all other players play strategies in $Y$ by $\bar Y_{-p}$.  Given a $p$-antiprofile $\bar q$, we write $x_{\bar q}$ as shorthand for $\prod_{j\neq p} x^j_{\bar q_j}$.

We will use a rearranged form of the replicator equation in several proofs. 
\begin{lem} \label{lem: other replicator form}
The replicator equation is equivalent to
\[
\dot x^p_s = x_s^p\sum_{t\in S_p} x_t^p \sum_{\bar q\in\bar Z_{-p}} x_{\bar q} \left (u_p(s;\bar q) - u_p(t; \bar q) \right )
\]
\end{lem}
\begin{proof}
\begin{align*}
    \dot x_s^p &= x_s^p\left(\exptu_p(s; \bar x_{-p}) - \sum_{t\in S_p} x_t^p \exptu_p(t; \bar x_{-p}) \right) \\
     &= x_s^p\left (\sum_{\bar q\in\bar Z_{-p}} x_{\bar q} u_p(s; \bar q) -  \sum_{t\in S_p} x_t^p \sum_{\bar q\in\bar Z_{-p}} x_{\bar q} u_p(t; \bar q)  \right ) \\
    &= x_s^p\sum_{t\in S_p} x_t^p \left(\sum_{\bar q\in\bar Z_{-p}} x_{\bar q} u_p(s; \bar q) -  \sum_{\bar q \in\bar Z_{-p}} x_{\bar q} u_p(t ;\bar q) \right) \\
    &= x_s^p\sum_{t\in S_p} x_t^p \sum_{\bar q\in\bar Z_{-p}} x_{\bar q} \left (u_p(s;\bar q) - u_p(t; \bar q) \right )
\end{align*}
where in the third equality we have used the fact that $\sum_{t\in S_p} x^p_t = 1$.
\end{proof}

\begin{thm}[Theorem~\ref{subgame invariance}]
Let $X$ be the strategy space of a game $u$, and $Y$ be the strategy space of a subgame $u'$ of $u$. The flow $\phi_u|_Y$ of the replicator on $u$ restricted to $Y$ is identical to $\phi_{u'}$, the replicator flow on $u'$.
\end{thm}
\begin{proof}
Let $x$ be a mixed profile in the subgame $Y$. Observe from the definition of the replicator that for any player $p$ with strategy $t$ outside of $Y$, $x^p_t = 0$ and so $\dot x^p_t = 0$. Now we show that if $x^p_t > 0$ then $\dot x^p_t$ depends on the payoffs of profiles in $Y$. Then note that $x_{\bar q} = 0$ if $\bar q\not \in \bar Y_{-p}$. Then, for $s\in Y_p$, by Lemma~\ref{lem: other replicator form},
\begin{align*}
    \dot x_s^p
    &= x_s^p\sum_{t\in S_p} x_t^p \sum_{\bar q\in\bar Z_{-p}} x_{\bar q} \left (u_p(s;\bar q) - u_p(t; \bar q) \right ) \\
    &= x_s^p\sum_{t\in Y_p} x_t^p \sum_{\bar q\in\bar Z_{-p}} x_{\bar q} \left (u_p(s;\bar q) - u_p(t; \bar q) \right ) \\
    &= x_s^p\sum_{t\in Y_p} x_t^p \sum_{\bar q\in\bar Y_{-p}} x_{\bar q} \left (u_p(s;\bar q) - u_p(t; \bar q) \right ) \\
\end{align*}
and this is the definition of the replicator on the restricted game $Y$.
\end{proof}

\begin{lem}
If $\chain{x}{y}$ then $\mathcal{A}_x\subseteq \mathcal{A}_y$ and $\mathcal{R}_y\subseteq \mathcal{R}_x$. Conversely, if $x$ is chain recurrent, then $\mathcal{A}_x \subseteq \mathcal{A}_y$ implies $\chain{x}{y}$, and $\mathcal{R}_x\subseteq \mathcal{R}_y$ implies $\chain{y}{x}$. Consequently, if $x$ is chain recurrent, $\chaineq{x}{y}$ if and only if $\mathcal{A}_x = \mathcal{A}_y$ and $\mathcal{R}_x = \mathcal{R}_y$. If $y$ is also chain recurrent then one of these equalities is sufficient.
\end{lem}
\begin{proof}
The first claim is well-known, with a proof in~\citep{akin1984evolutionary}.

For the second claim, assume $x$ is chain recurrent and $\mathcal{A}_x\subseteq \mathcal{A}_y$. Now fix some $\epsilon > 0$ and $T > 0$, and let $C_{\epsilon,T}(x)$ be the set of points reachable by $(\epsilon,T)$-chains from $x$. Note that $C_{\epsilon,T}(x)$ is open. If $C_{\epsilon,T}(x)$ is all of $X$, then $y\in C_{\epsilon,T}(x)$. Otherwise, consider $V = N_{\epsilon/2}(C_{\epsilon,T}(x))$, that is the set of points within $\epsilon/2$ distance of $C_{\epsilon,T}(x)$. This is open. Now let $z$ be some point within $\epsilon$ of a point $\chi\in C_{\epsilon,T}(x)$. By definition, $\phi(z,t) \in C_{\epsilon,T}(x)$ for all $t\geq T$, as we can obtain a $(\epsilon,T)$-chain to it by adding a step to the chain from $x$ to $\chi$. Thus, we find that $\phi(\cl(V),t) \subseteq \phi(N_\epsilon(C_{\epsilon,T}(x)),t) \subseteq C_{\epsilon,T}(x)\subset V$. Hence $V$ is a trapping region for some attractor $A\subseteq C_{\epsilon,T}(x) \subseteq V$. Let $A^*$ be the dual repellor of $A$. By Theorem~\ref{recurrence characterisation}, $x\in A\cup A^*$, and since $x\in V$ and $A^*\subseteq X\setminus V$ we have $x\in A$. But then $y\in A$, since $y$ is contained in all attractors containing $x$, and so $y\in A \subseteq C_{\epsilon,T}(x)$. Since $\epsilon$ and $T$ were arbitrary, $\chain{x}{y}$. Now observing that in the time-reversed flow $\phi^{-1}$, repellors become attractors and vice versa, and pseudo-orbits $\chain{x}{y}$ become pseudo-orbits $\chain{y}{x}$, establishing the claim for repellors.

If $\chaineq{x}{y}$, then $\chain{x}{y}$ and $\chain{y}{x}$, so $\mathcal{A}_x = \mathcal{A}_y$ and $\mathcal{R}_y = \mathcal{R}_x$ by the first claim. By the second claim, if $x$ is chain recurrent and $\mathcal{A}_x = \mathcal{A}_y$ and $\mathcal{R}_y = \mathcal{R}_x$ then $\chaineq{x}{y}$, and so $y$ is necessarily also chain recurrent. If we know in advance that $y$ is chain recurrent, then one of the equalities is sufficient, because by Theorem~\ref{recurrence characterisation} if two chain recurrent points are contained in all the same attractors then they are also contained in all the same repellors.
\end{proof}

We will need the following result.

\begin{lem}[\cite{kalies2021lattice}] \label{intersection of attractors}
If $\phi$ is a flow, then any finite non-empty intersection of attractors is an attractor.
\end{lem}

\begin{thm}[Theorem~\ref{scc theorem}]
Let $x\in X$ be a chain recurrent point under the replicator. Then there exists a pure profile $y$ and pseudo-orbit $\chain{x}{y}$ where $y$ is in a sink connected component $H$ contained in a sink chain component $[H]$.
\end{thm}
\begin{proof}
First, we define a preorder on the set of sink connected components of a game, defined by the attractors of the replicator. If $K$ and $H$ are sink connected components, we say $K \preceq H$ if $\mathcal{A}_K \supseteq \mathcal{A}_H$. This is a preorder because it inherits reflexivity and transitivity from the sets $\mathcal{A}_K$ and $\mathcal{A}_H$. By finiteness, this order has at least one minimal element, and we denote the minimal elements by $M_1,\dots,M_m$. We will show that the sink chain components are exactly the chain components $[M_1],\dots,[M_m]$ (which are not necessarily distinct).

(\emph{Claim:} Every chain recurrent point has a pseudo-orbit to some $M_i$.)

Let $x$ be chain recurrent. Now define 
\[
\mathcal{M}_x := \bigcap_{A \in \mathcal{A}_{\{x\}}} \{\text{sink connected components contained in }A\}
\]
Recall that each such $A$ contains a sink connected component (Lemma~\ref{attractors contain nodes}). Suppose for contradiction that $\mathcal{M}_x$ is empty. Then for each sink connected component $H_1,\dots,H_n$ in the game there exist attractors $B_1,\dots,B_n$, each containing $x$, where $B_i$ does not contain $H_i$. By Lemma~\ref{intersection of attractors}, $\bigcap_{i=1}^n B_i$ is an attractor, which by Lemma~\ref{attractors contain nodes} contains a sink connected component $H_i$. This is a contradiction, because $H_i$ is not contained in $B_i$, by construction. Hence $\mathcal{M}_x$ is not empty. Let $K$ be a sink connected component in $\mathcal{M}_x$. By definition of $\mathcal{M}_x$, $\mathcal{A}_{\{x\}}\subseteq \mathcal{A}_K$.

There exists some sink connected component $M_i$ with $M_i \preceq K$ and so $\mathcal{A}_{\{x\}}\subseteq \mathcal{A}_K \subseteq \mathcal{A}_{M_i}$, and by Lemma~\ref{chains in attractors} there is a pseudo-orbit $\chain{x}{y}$ from $x$ to every point $y$ in $M_i$.

(\emph{Claim:} The $[M_i]$s are sink chain components.)

Let $z$ be some chain recurrent point, and suppose there is a pseudo-orbit $\chain{x}{z}$, with $x\in [M_i]$. By the above argument there is some $M_j$ with $\chain{z}{y}$, $y\in M_j$. By transitivity there is a pseudo-orbit $\chain{x}{y}$, which by Lemma~\ref{chains in attractors} implies that $\mathcal{A}_{M_i}\subseteq \mathcal{A}_{M_j}$, but since each $M_i$ is minimal in this order we must have $\mathcal{A}_{M_i} = \mathcal{A}_{M_j}$. Lemma~\ref{chains in attractors} establishes that $[M_i] = [M_j]$. Hence $z\in [M_i]$ and so $[M_i]$ is a sink chain component.
\end{proof}

\begin{thm}[Theorem~\ref{content containment}]
Let $H$ be a connected component of the response graph. Then the content of $H$ is contained in $[H]$, the chain component of the replicator containing $H$.
\end{thm}
\begin{proof}
We show that if all of the pure profiles in a game $Y$ are in the same chain component $K$, then all profiles in the game are in $K$. We will use induction, by Theorem~\ref{subgame invariance}. First, observe that the base case where $Y$ is a pure profile is trivial.

Now suppose for induction that all mixed profiles in all proper subgames of $Y$ are contained in $K$. The union of all proper subgames is the boundary of $Y$, so the boundary of $Y$ is contained in $K$. Suppose there is an attractor $A$ in $Y$. By Theorem~\ref{no interior attractors}, $A$ intersects the boundary, but then by Lemma~\ref{chains in attractors}, $A$ must contain all of the boundary, since all points on the boundary are chain equivalent. But then the dual repellor $A^*$ is contained in the interior of $Y$, which is impossible by Theorem~\ref{no interior attractors}. We conclude that there are no attractors or repellors in $Y$, and so all points are chain equivalent by Lemma~\ref{chains in attractors}.

Finally, we know by Lemma~\ref{graph reachable} that all profiles in a connected component $H$ are chain recurrent and in the same chain component $[H]$. By the above argument, in any subgame where all pure profiles are in $H$, all mixed profiles are chain equivalent, and so in $[H]$. The content of $H$ is the union of all of these subgames.
\end{proof}

\begin{lem}[Lemma~\ref{attracting subgames}]
If $Y$ is an attracting subgame, then the set of mixed profiles in $Y$ is an attractor under the replicator.
\end{lem}
\begin{proof}
Given a number $0<M<1$ we define $\mathcal{Y}_M$ to be the $M$-neighbourhood of $Y$ in the infinity norm; that is, the set of mixed profiles $x$ where for any player $p$ and strategy $s\in S_p\setminus Y_p$, $x^p_s < M$.

We know that the points in $Y$ form a compact non-empty invariant set, by Theorem~\ref{subgame invariance}, and we assume that $Y$ is not the whole game.

(\emph{Claim:} there exists an $M > 0$ such that in $\mathcal{Y}_M$, $\dot x^p_r < 0$ for any strategy $s\in S_p\setminus Y_p$)

Fix a player $p$, and let $s\in S_p \setminus Y_p$. By Lemma~\ref{lem: other replicator form},
\[
\dot x^p_s = x_s^p\sum_{t\in S_p} x_t^p \sum_{\bar q\in\bar Z_{-p}} x_{\bar q} \left (u_p(s;\bar q) - u_p(t; \bar q) \right )
\]
Now, if $t\in S_p$, we define $\alpha^p_{s,t} = \max_{\bar q\in\bar Y_{-p}}( u_p(s;\bar q) - u_p(t;\bar q))$ and $\beta^p_{s,t} = \max_{\bar q\in\bar Z_{-p}\setminus \bar Y_{-p}} (u_p(s;\bar q) - u_p(t;\bar q))$. As $Y$ is attracting and $t\in S_p$, $\alpha^p_{s,t}$ must be negative. Finally, define $\gamma^p_s = \max_{t\in S_p\setminus Y_p} \max_{\bar q\in \bar Z_{-p}} ( u_p(s;\bar q) - u_p(t;\bar q)) $.

Then
\begin{align*}
    \dot x^p_s &= x_s^p\sum_{\bar q \in\bar Z_{-p}} x_{\bar q} \sum_{t\in S_p} x_t^p  \left (u_p(s; \bar q)_p - u_p(t;\bar q)\right ) \\
    &= x_s^p\sum_{\bar q \in\bar Z_{-p}} x_{\bar q} \left ( \sum_{t\in Y_p} x_t^p \left (u_p(s; \bar q) - u_p(t; \bar q)\right ) + \sum_{t\in S_p\setminus Y_p} x_t^p \left (u_p(s; \bar q) - u_p(t; \bar q)\right ) \right) \\
    &\leq x_s^p\left ( \sum_{t\in Y_p} x_t^p \left (\alpha^p_{s,t} \sum_{\bar q\in \bar Y_{-p}}  x_{\bar q} + \beta^p_{s,t} \sum_{\bar q\in \bar Z_{-p}\setminus \bar Y_{-p}}  x_{\bar q} \right) + \gamma^p_s I(\gamma^p_s > 0)\sum_{t\in S_p\setminus Y_p} x_t^p \right )
\end{align*}
where $I(\cdot)$ is the indicator function. 
Now observe that every antiprofile $\bar q\in \bar Z_{-p}\setminus \bar Y_{-p}$ contains at least one strategy not in $Y$, and so for $x\in \mathcal{Y}_M$, $x_{\bar q} \leq M$. Thus $\sum_{\bar q\in \bar Z_{-p}\setminus \bar Y_{-p}} x_{\bar q} \leq M Q$, where $Q$ is the cardinality of $Z_{-p}\setminus \bar Y_{-p}$. Hence $\sum_{\bar q\in \bar Y_{-p}} x_{\bar q} \geq 1-MQ$. For $M < 1/Q$, $1-MQ > 0$. Given $\alpha^p_{s,t} < 0$, we have $\alpha^p_{s,t} \sum_{\bar q\in \bar Y_{-p}} x_{\bar q} \leq \alpha^p_{s,t} (1-MQ)$ and $\beta^p_{s,t} I(\beta^p_{s,t} > 0) \sum_{\bar q\in \bar Z_{-p}\setminus \bar Y_{-p}}  x_{\bar q} \leq MQ$. Also, for $x\in \mathcal{Y}_M$ we have $\sum_{t\in S_p\setminus Y_p} x_t^p \leq MN$. Hence
\begin{align*}
    \dot x^p_s &\leq x_s^p\left ( \sum_{t\in Y_p} x_t^p \left (\alpha^p_{s,t} (1-MQ) + \beta^p_{s,t} MQ \right) + \gamma^p_s I(\gamma^p_s > 0)NM \right ) \\
    &\leq x_s^p\left ( \sum_{t\in Y_p} x_t^p \alpha^p_{s,t} - MQ  \sum_{t\in Y_p} \alpha^p_{s,t} + MQ \sum_{t\in Y_p} \beta^p_{s,t}  + \gamma^p_s I(\gamma^p_s > 0)NM \right)
\end{align*}
Finally, let $\alpha^p_s = \max_{t\in Y_p} \alpha^p_{s,t}$. Then $\sum_{t\in Y_p} x_t^p \alpha^p_{s,t} \leq \alpha^p_s \sum_{t\in Y_p} x^p_t \leq \alpha_s^p(1 - M L_p)$, where $L_p = |S_p\setminus Y_p|$, and recalling that $\alpha_s^p < 0$. Thus we obtain the bound:
\begin{align*}
    \dot x^p_s &\leq x_s^p\left ( \alpha^p_s(1-ML_p) - MQ  \sum_{t\in Y_p} \alpha^p_{s,t} + MQ \sum_{t\in Y_p} \beta^p_{s,t}  + \gamma^p_s I(\gamma^p_s > 0)NM \right) \\
    &\leq x_s^p\left ( \alpha^p_s +M\left (-L_p\alpha^p_s - Q  \sum_{t\in Y_p} \alpha^p_{s,t} + Q \sum_{t\in Y_p} \beta^p_{s,t}  + \gamma^p_s I(\gamma^p_s > 0)N \right )\right) \\
\end{align*}
The sign of this equation is determined by the sign of the bracketed term, which is a linear equation in $M$ that is independent of $x$. Given $\alpha_s^p < 0$, there is a positive value $M^p_s$ ($< 1/Q)$) such that for any $M \leq M^p_s$ we have $\dot x^p_s$ always negative. Finally, define $M := \min_p \min_{s\in S_p\setminus Y_p} M^p_s$, which is positive as there are only finitely many strategies. This establishes the claim.

Since $x^p_s$ is strictly decreasing inside $\mathcal{Y}_M$, this set is forward-invariant and so $\phi^t(x)^p_s \to 0$ as $t\to\infty$. The result then follows from
\begin{align*}
    \lim_{t\to\infty} \sup_{x\in \mathcal{Y}_M} \min_{y\in Y} \dist(\phi^t(x),y)
    &= \lim_{t\to\infty} \sup_{x\in \mathcal{Y}_M} \min_{y\in Y} \|\phi^t(x) - y\|_\infty \\
    &= \lim_{t\to\infty} \sup_{x\in \mathcal{Y}_M} \min_{y\in Y} \max_p \max_{s\in S_p} |\phi^t(x)^p_s - y^p_s | \\
    &= \lim_{t\to\infty} \sup_{x\in \mathcal{Y}_M}  \max_p \max_{s\in S_p} \min_{y\in Y} |\phi^t(x)^p_s - y^p_s |
\end{align*}
The value $\min_{y\in Y} |\phi^t(x)^p_s - y^p_s |$ is zero if the strategy $s\in S_p$ is inside $Y$. Consequently, since $\phi^t(x)$ is not in $Y$, the strategy $s$ which maximises \[ \max_p \max_{s\in S_p} \min_{y\in Y} |\phi^t(x)^p_s - y^p_s | \] must be outside $Y$. Thus $\phi^t(x)^p_s \to 0$ as $t\to \infty$, for any $x\in \mathcal{Y}_M$, and so
\[
\lim_{t\to\infty} \sup_{x\in \mathcal{Y}_M}  \max_p \max_{s\in S_p} \min_{y\in Y} |\phi^t(x)^p_s - y^p_s |  = 0
\]
\end{proof}

\begin{lem} \label{2xn connected}
The sink connected components of strict $2\times n$ games are subgames. If the game is preference-zero-sum game without dominated strategies, then the response graph is strongly connected.
\end{lem}
\begin{proof}
We can assume there are no dominated strategies, because the sink connected components are always contained within the subgame surviving iterated dominance. We call the strategy sets $\{A,B\}$ and $\{s_1,\dots,s_n\}$ respectively. Let $s_1,s_2,\dots,s_n$ be the order in which player 2 prefers their strategies when player 1 plays $A$. Because no strategy dominates any other, when player 1 plays $B$ player 2 must prefer their own strategies in the reverse order, which is $s_n,s_{n-1},\dots,s_1$. Thus there are arcs $\arc{(A,s_i)}{(A,s_{i+1})}$ and $\arc{(B,s_i)}{(B,s_{i-1})}$ in the response graph for each $s_i$.

Any non-singleton sink connected component must contain profiles with both of player 1's strategies $A$ and $B$. Without loss of generality suppose that $(A,s_i)$ and $(B,s_j)$ are in a sink connected component. By the structure of the graph, there are nodes $(B,s_k)$ and $(A,s_\ell)$ with $k\geq n$, $\ell \leq i$ and paths from $(A,s_i)$ and $(B,s_j)$ respectively. All nodes reachable from $(B,s_k)$ and $(A,s_\ell)$ are also in the sink connected component, and so in particular $(A,s_j)$ and $(B,s_i)$ are both in the component. This is true for any pair of profiles, so the component is a subgame.

Now we consider the second claim. There are two cases: (1) either one of $(B,s_1)$ or $(A,s_n)$ is a sink, or (2) there are arcs $\arc{(B,s_1)}{(A,s_1)}$ and $\arc{(A,s_n)}{(B,s_n)}$ in the graph. In case (2) there is a Hamiltonian cycle
\[ (A,s_1),(A,s_2),\dots,(A,s_n),(B,s_n),(B,s_{n-1}),\dots,(B,s_1),(A,s_1)
\] containing all nodes so the graph is strongly connected. We show case (1) is impossible in a preference-zero-sum game. Assume without loss of generality that $(A,s_n)$ is a sink, so there is an arc $\arc{(B,s_n)}{(A,s_n)}$. As $A$ does not dominate $B$, there is an $s_i$ with $\arc{(A,s_i)}{(B,s_i)}$. But then the $2\times 2$ subgame $\{A,B\}\times \{s_i,s_n\}$ has the response graph of $2\times 2$ Coordination, which is impossible in a preference-zero-sum game, by~\cite[Corollary~4.7]{biggar2022twoplayer}.
\end{proof}

\begin{lem} \label{lem: diamond games}
No zero-sum game with an interior equilibrium has a response graph isomorphic to the Inner or Outer Diamond graphs.
\end{lem}
\begin{proof}
Suppose for contradiction that $p$ is an interior Nash equilibrium of a game $u = (u_1,u_2)$ with the response graph isomorphic to the Inner Diamond. Then by~\cite{papadimitriou2016nash}, all profiles must be chain recurrent. However by Lemma~\ref{attracting subgames} the  sink of the Inner Diamond graph is an attractor, so cannot be chain equivalent to all other points by Lemma~\ref{chains in attractors}. For the Outer Diamond graph, observe similarly that the source is a repellor, which again cannot be chain equivalent to all other points, by Lemma~\ref{chains in attractors}.
\end{proof}

\end{document}